\documentclass[lettersize,journal,12pt,onecolumn]{IEEEtran}
\usepackage{amsmath,amsfonts}
\usepackage{algorithmic}
\usepackage{algorithm}
\usepackage{amssymb}
\usepackage{array}
\usepackage[caption=false,font=normalsize,labelfont=sf,textfont=sf]{subfig}
\usepackage{textcomp}
\usepackage{stfloats}
\usepackage{url}
\usepackage{xcolor}
\usepackage{verbatim}
\usepackage{graphicx}
\usepackage{cite}
\usepackage{stfloats} 
\usepackage{caption}  
\newtheorem{definition}{Definition}
\newtheorem{theorem}{Theorem}
\newtheorem{corollary}{Corollary}

\newtheorem{remark}{Remark}

\begin{document}
\title{Theoretical Analyses of Detectors for Additive Noise Channels with Mean-Variance Uncertainty under Nonlinear Expectation Theory}
	
\author{Wen-Xuan Lang, Guiying Yan, Zhi-Ming Ma

\thanks{The authors are with the National Center for Mathematics and Interdisciplinary Sciences (NCMIS) and Academy of Mathematics and Systems Science, CAS, Beijing 100190, China (e-mail: langwx@amss.ac.cn; yangy@amss.ac.cn; mazm@amt.ac.cn).}
}


		
\maketitle

\begin{abstract}
In classical information theory, both the form and performance of the optimal detector for additive noise channels can be precisely derived, based on the assumption that the channel noise follows a specific probability distribution or a mixture of known distributions, or that the exact distribution exists but is unknown. In this paper, we extend the analyses of detectors for additive noise channel to the situation where the probability model for analyzing channels is uncertain, utilizing nonlinear expectation theory. We consider two types of distribution uncertainties: one with no mean uncertainty but with variance uncertainty, and another with both mean and variance uncertainties. We derive the optimal threshold detectors for binary input additive noise channel under the nonlinear expectation optimal criterion for both scenarios and provide their explicit forms. Our findings reveal that mean uncertainty significantly influences the form of the optimal detector, whereas variance uncertainty does not. Additionally, we propose an estimation method for the uncertain parameters of the channel noise. Finally, we present theoretical analyses and simulated performance results of the newly derived optimal threshold detectors, and compare these results with the performance of optimal detector under classical information theory, which assumes a deterministic probability model. The results of experiments show that our new detection methods outperform conventional methods in most scenarios with uncertain probability models, showing the practical relevance of our theoretical contributions.

\end{abstract}


\section{Introduction}\label{sec1}

In classical information theory, the additive noise channels are useful and common for some communication channels, such as wireless telephone channels and satellite links \cite{thomas2006elements}. Reliable communication over additive noise channels is a fundamental topic in information theory, studied extensively since the 1950s \cite{shannon1948mathematical,6767457}. The fundamental detection problem in communication channels is accurately identifying the characteristics of signals within a channel to ensure reliable data transmission. The analyses of the detection problem over a Gaussian channel, and the pairwise point error probability, which is the probability of the received point being closer to another point than to the sent point, can be seen in \cite{boutros1998signal,tse2005fundamentals,forney1998modulation}.

Uncertainty inherent in the channel is a key factor in wireless communication, which makes the field both challenging and interesting \cite{tse2005fundamentals,rak2015future}. Characterizing this uncertainty is crucial for understanding and potentially enhancing the reliability of these systems, particularly in the presence of additive noise. In various communication environments, the distributions of noise should be considered as uncertain due to temperature changes, environment noise variations, calibration errors and so on \cite{rojas2022snr}. Many studies have also considered the impact of uncertain noise distribution, addressing various aspects such as covert communication, energy detection, channel estimation, and performance within communication systems \cite{ta2019covert,sobers2017covert,kortun2012throughput,mahendru2019adaptive,gaiera2019performance,dos2024influence}. This uncertain distributions noise can be significant in many applications and must be taken into consideration to model system performance. However, in the previous works, researchers acknowledge the existence of a particular distribution characterizing the channel, even though the specific parameters or form of this distribution are uncertain. Consequently, the probability models in the existing works are assumed to have no uncertainty.

In contrast, for channels with uncertain probability models, the detection problem mentioned in this paper remains largely unexplored in the existing research field. In the complex physical world fraught with a significant amount of unanticipated and nonstationary random phenomena, the assumption of precise and well-defined probability distributions to describe uncertainties is somewhat idealized. This reality challenges the traditional approach of relying on deterministic probability models, which typically assume the existence of an underlying ``true'' probability distribution. Bias or error may be imposed on signal or information processing algorithms once the real distribution of the information source diverges from the estimated distribution.

Nonlinear expectation theory \cite{peng2007g}, an emerging field of research that extends classical probability theory by using sublinear expectations instead of linear expectations, provides a powerful framework for effectively addressing the uncertainty of probability models. Based on the representation theorem, sublinear expectation can be regarded as a family of probability distributions. By incorporating the distributions uncertainty, this theory provides a more robust framework for modeling complex systems. One important result of the theory is \textit{the central limit theorem with law of large numbers} under sublinear expectations \cite{peng2008new,peng2019law}. This theorem indicates that the maximal distribution and the G-normal distribution are highly universal distributions and widely applicable, and can be regarded as a generalization of the Gaussian distribution in probability theory. By integrating these two distributions, one can achieve a comprehensive characterization of the random variable that accounts for both mean and variance uncertainties. Furthermore, these two distributions have offered significant practical advantages in fields such as finance and risk management \cite{peng2023improving,yang2023linear}, where understanding the interplay between mean and variance uncertainties is crucial for effective decision-making and modeling. In the recent monograph \cite{peng2019nonlinear}, this theory has been fully elaborated.

In this paper, we demonstrate that the nonlinear expectation theory can be utilized to analyze the performance of detection problem in additive noise channel where the noise distributions are uncertain. An important component of this work is the focus on the distributions uncertainty of the channel noise. We analyze two types of distribution uncertain cases: one is that the noise random variable has no mean uncertainty but has variance uncertainty, and another is that the noise random variable has both mean and variance uncertainties. We utilize a combination of the maximal distribution and the G-normal distribution to characterize channel noise with mean and variance uncertainties. We derive the optimal detectors for binary input signals in each of these two scenarios and provide theoretical analyses of these systems.

We denote $X(t)$, $Y(t)$ and $Z(t)$ as input, output and noise random variable at time $t$, respectively. Let
\begin{equation}\label{additive noise channel}
Y(t)=X(t)+Z(t).
\end{equation}
There are three characteristics that make our channel model fundamentally different from the one in classical information theory. First, and most importantly, the distribution of the noise $Z(t)$ is uncertain, which is the central concept addressed by using nonlinear expectation theory. Second, the noise $Z(t)$ is independent to the input $X(t)$ under sublinear expectation $\mathbb{E}$. Third, for different times $t_1 < t_2$, the noise $Z(t_2)$ is independent and identically distributed (abbreviation: IID) to $Z(t_1)$ under sublinear expectation $\mathbb{E}$.

The main contributions can be summarized as follows:
\begin{itemize}
	\item[$\bullet$] We study the case when noise random variable with both mean and variance uncertainties, and propose the concept of the \textit{ADUN channel}. For binary input signals $\{x_A,x_B\} \subset \mathbb{R}, x_A>x_B$ transmitted over this channel, we derive the uncertain family of probability distributions of channel output. Furthermore, we give the necessary and sufficient condition for the existence of the optimal threshold detector. If exist, the optimal threshold detector is to choose $x_A$ when the received signal is greater than $\frac{x_A+x_B+\overline{\mu}+\underline{\mu}}{2}$ and $x_B$ otherwise. Compared with the traditional minimum distance detector, we find that our optimal threshold detector outperforms conventional method in most cases where the mean of noise has uncertainty. We also provide a method for estimating the range of mean uncertainty and variance uncertainty of channel noise.
	
	\item[$\bullet$] We consider the case when noise random variable with no mean uncertainty but with variance uncertainty, and propose the concept of the \textit{ADUN channel without mean uncertainty}. Given the binary input signals $\{x_A, x_B\} \subset \mathbb{R}, x_A>x_B$ are transmitted across this channel, we also derive the uncertain family of probability distributions of channel output. Furthermore, we obtain that the minimum distance rule, which is shown as choosing $x_A$ when the received signal is greater than $\frac{x_A+x_B}{2}$ and $x_B$ otherwise, is the optimal detector in this case. We also study the detection in the presence of both Rayleigh channel fading and communication noise. We provide sufficient conditions under which the detection problem in fading channels can be effectively transformed into that of additive noise channels when the receiver has knowledge of the channel fading.
\end{itemize}

Being the first to consider the impact of the inherent uncertainty in probability models on specific scenarios in wireless communication, we develop novel approaches to analyze the performance of channels. This paper develops the application area of sublinear expectation. Based on our current understanding, this study pioneers research into detection in additive noise channels under nonlinear expectation theory.

The remainder of the paper is organized as follows. In Section \ref{sec2}, we recall some basic results in nonlinear expectation theory and the corresponding notions, as well as the detection problem in additive noise channels in classical information theory. In Section \ref{sec3}, we establish a model for additive noise channels with distribution uncertainty, and propose the concepts of \textit{ADUN channel} and \textit{ADUN channel without mean uncertainty}, which play crucial roles in our analyses. In Section \ref{sec4}, we derive the optimal threshold detectors under the nonlinear expectation optimality criterion for both ADUN channels and ADUN channels without mean uncertainty when employing binary input signals. Additionally, we propose estimation methods for the uncertain parameters of the channel noise. In Section \ref{sec5}, we provide sufficient conditions that allow for the effective conversion of the detection problem in fading channels into the problem that in additive noise channels, assuming the receiver has knowledge of the channel fading. The theoretical and simulated performance results of the newly derived optimal detectors are shown in Section \ref{sec6}. Finally, we draw conclusions in Section \ref{sec7}.

\section{Preliminaries}\label{sec2}

\subsection{Basic notions of nonlinear expectation theory}\label{subsec2.1}

Nonlinear expectation theory is a novel axiomatic system established by Peng \cite{peng2007g}. The central idea of this theory is that random variables have distributions uncertainty. In this subsection, we provide a concise overview of some useful definitions and theorem from the nonlinear expectation theory. For more details, researchers can refer to \cite{peng2019nonlinear,peng2017theory}.

Let $\Omega$ be a given set. Let $\mathcal{H}$ be a linear space of real valued functions defined on $\Omega$. Suppose that $\mathcal{H}$ satisfies the following three conditions:

1. $c\in \mathcal{H}$ for each constant $c \in \mathbb{R}$.

2. $|X|\in \mathcal{H}$ if $X \in \mathcal{H}$.

3. $\phi(X_1,\cdots, X_n)\in \mathcal{H}$ for each $\phi \in \mathbb{L}^{\infty}(\mathbb{R}^n)$ if $X_1,\cdots, X_n\in \mathcal{H}$.

\noindent Here $\mathbb{L}^{\infty}(\mathbb{R}^n)$ denotes the space of bounded Borel-measurable functions. The space $\mathcal{H}$ will be used as the space of random variables.

\begin{definition}
	A \textit{sublinear expectation} $\mathbb{E}:\mathcal{H} \rightarrow \mathbb{R}$ is a functional defined on the space $\mathcal{H}$ satisfying the following properties:
	\begin{enumerate}
		\item{\textit{Monotonicity}: $\mathbb{E}[X]\geq \mathbb{E}[Y]$, if $X\geq Y$.}
		\item{\textit{Constant preserving}: $\mathbb{E}[c]=c,\forall c \in \mathbb{R}$.}
		\item{\textit{Sub-additivity}: $\mathbb{E}[X+Y]\leq \mathbb{E}[X]+\mathbb{E}[Y]$, $\forall X,Y \in \mathcal{H}$.}
		\item {\textit{Positive homogeneity}: $\mathbb{E}[\lambda X]=\lambda \mathbb{E}[X]$, for $\lambda>0$.}
	\end{enumerate}
	The triplet $(\Omega,\mathcal{H},\mathbb{E})$ is called a \textit{sublinear expectation space}. If $\mathbb{E}$ satisfies only 1 and 2, then $\mathbb{E}$ is called a \textit{nonlinear expectation} and $(\Omega,\mathcal{H},\mathbb{E})$ is called a \textit{nonlinear expectation space}.
\end{definition}

Given a sublinear expectation $\mathbb{E}$, the conjugate expectation $\mathcal{E}$ of sublinear expectation $\mathbb{E}$ is defined as
\begin{equation}
\mathcal{E}[X]:=-\mathbb{E}[-X],\quad \forall X\in \mathcal{H}.
\end{equation}

\begin{theorem}
	Let $\mathbb{E}$ be a sublinear expectation defined on $(\Omega,\mathcal{H})$. Then there exists a family of linear functionals $\{E_{\theta}:\theta \in \Theta \}$ defined on $\mathcal{H}$, such that
	\begin{equation}
	\mathbb{E}[X]=\sup_{\theta \in \Theta }E_{\theta}[X].
	\end{equation}
	If $\mathbb{E}$ is also continuous from above, i.e. $\mathbb{E}[X_{i}]\downarrow 0$ for each sequence $\{ X_{i}\}_{i=1}^{\infty}$ of random variables in $\mathcal{H}$ such that $X_{i}\downarrow 0$, then for each $\theta \in \Theta$, there exists a probability measure $P_{\theta}$ defined on the measurable space $(\Omega,\sigma(\mathcal{H}))$ such that
	\begin{equation}
	E_{\theta}[X]=\int_{\Omega} XdP_{\theta}, \quad \forall X\in \mathcal{H}.
	\end{equation}
	Then $E_{\theta}$ is also denoted as $E_{P_{\theta}}$ and $\mathbb{E}[\cdot]=\underset{\theta \in \Theta}{\sup} E_{P_{\theta}}[\cdot]$. Here ``$ X_i \downarrow 0$ " means $X_i$ monotonically decreases to zero, and $\sigma(\mathcal{H})$ is the smallest $\sigma$-algebra generated by $\mathcal{H}$.
\end{theorem}

\begin{definition}
	Let $X_1$ and $X_2$ be two random variables defined on sublinear expectation spaces $(\Omega,\mathcal{H},\mathbb{E}_1)$ and $(\Omega,\mathcal{H},\mathbb{E}_2)$, respectively. They are called \textit{identically distributed}, denoted by $X_1\overset{d}{=}X_2$, if for any $\phi \in \mathbb{L}^{\infty}(\mathbb{R})$,
	\begin{equation}
	\mathbb{E}_1[\phi(X_1)]=\mathbb{E}_2[\phi(X_2)].
	\end{equation}
\end{definition}

\begin{definition}\label{definition:independent}
	In a sublinear expectation space $(\Omega,\mathcal{H},\mathbb{E})$, a random variable $Y$ is said to be \textit{independent} of another random variable $X$ under $\mathbb{E}[\cdot]$ if for any $\phi \in \mathbb{L}^{\infty}(\mathbb{R}^2)$,
	\begin{equation}\label{eq:independent}
	\mathbb{E}[\phi(X,Y)]=\mathbb{E}[\mathbb{E}[\phi(x,Y)|x=X]].
	\end{equation}
\end{definition}

\begin{definition}
	The random variable $\overline{X}$ is said to be an \textit{independent copy} of $X$ if $\overline{X}$ is independent of $X$ and $\overline{X}\overset{d}{=}X$.
\end{definition}

In this new framework, the nonlinear versions of the notions of independence and identical distribution play crucial roles.

\begin{definition}\label{maximal distribution}
	A $n$-dimensional random variable $\eta$ defined on a sublinear expectation space $(\Omega,\mathcal{H},\mathbb{E})$ is called \textit{maximally distributed} if there exists a bounded, closed, and convex subset $\Gamma \subset \mathbb{R}^n$ such that for any $\varphi \in C_{l,lip}(\mathbb{R}^n)$
	\begin{equation}
	\mathbb{E}[\varphi(\eta)]=\sup_{x\in \Gamma}\varphi(x).
	\end{equation}
	We denote $\eta \overset{d}{=} M_{\Gamma}$. The distribution of $\eta$ is called \textit{maximal distribution}.
	
	Specifically, when $n=1$, there is $\Gamma = [\underline{\mu},\overline{\mu}]$, where $\underline{\mu}=-\mathbb{E}[-\eta]$, $\overline{\mu} = \mathbb{E}[\eta]$. Then we have $\eta \overset{d}{=} M_{[\underline{\mu},\overline{\mu}]}$ and
	\begin{equation}
	\mathbb{E}[\varphi(\eta)]=\sup_{x\in [\underline{\mu},\overline{\mu}]}\varphi(x).
	\end{equation}
\end{definition}

\begin{definition}\label{G-normal}
	A random variable $\delta$ defined on a sublinear expectation space $(\Omega,\mathcal{H},\mathbb{E})$ is called \textit{G-normally distributed} if
	\begin{equation}
	a\delta+b\overline{\delta} \overset{d}{=} \sqrt{a^2+b^2}\delta, \quad a,b\geq 0,
	\end{equation}
	where $\overline{\delta}$ is an independent copy of $\delta$. The distribution of $\delta$ is called \textit{G-normal distribution}.
\end{definition}

\begin{remark}\label{Remark:G-normal}
	For G-normally distributed $\delta$, there is $\mathbb{E}[\delta]=-\mathbb{E}[-\delta]=0$. Therefore, $\delta$ has no mean-uncertainty.
\end{remark}

The following Theorem \ref{G-normal and PDE-1} reveals a deep and essential relation between the probability under uncertainty and second order fully nonlinear parabolic equations.

\begin{theorem}\label{G-normal and PDE-1}
	A random variable $\delta$ is G-normally distributed with lower variance $\underline{\sigma}^2$ and upper variance $\overline{\sigma}^2$ ($0\leq \underline{\sigma} \leq \overline{\sigma}$), denoted by $\delta \overset{d}{=} N(0,[\underline{\sigma}^2,\overline{\sigma}^2])$, if for every $\varphi \in C_{l,lip}(\mathbb{R}^n)$, the function
	\begin{equation}
	u(t,x) := \mathbb{E}[\varphi(x+\sqrt{t}\delta)], (t,x) \in [0,\infty)\times R ,
	\end{equation} 
	is the unique viscosity solution of the following Cauchy problem
	\begin{equation}
	\left\{
	\begin{array}{ll}
	\partial_t u - G(\partial_{xx}^2 u) = 0  \\
	u(0,x)=\varphi(x)
	\end{array}
	\right. ,
	\end{equation}
	where $G(a)=\frac{1}{2}[\overline{\sigma}^2 (a)^{+} - \underline{\sigma}^2 (a)^{-}]$. In this case, there holds $\overline{\sigma}^2=\mathbb{E}[\delta^2]$ and $\underline{\sigma}^2=-\mathbb{E}[-\delta^2]$.
\end{theorem}

Let $\{\xi_k\}_{k\geq 1}$ be a sequence of IID random variables under sublinear expectation $\mathbb{E}$, $\underline{\mu} = -\mathbb{E}[-\xi_1]$ and $\overline{\mu}=\mathbb{E}[\xi_1]$. $\{\zeta_k\}_{k\geq 1}$ be another sequence of IID random variables under sublinear expectation $\mathbb{E}$ with $-\mathbb{E}[-\zeta_1]=\mathbb{E}[\zeta_1]=0$ and $\underline{\sigma}^2 = -\mathbb{E}[-\zeta_{1}^{2}]$, $\overline{\sigma}^2 = \mathbb{E}[\zeta_{1}^{2}]$. The main result of \textit{the central limit theorem with law of large numbers} under sublinear expectations is
\begin{equation}
\lim_{n\rightarrow \infty} \mathbb{E}\left[\varphi\left(\sum_{k=1}^{n}\left(\frac{\xi_k}{n}+\frac{\zeta_k}{\sqrt{n}}\right)\right)\right] = \mathbb{E}[\varphi(\eta+\delta)],
\end{equation}
where $\eta$ is maximally distributed and denoted by $M_{[\underline{\mu},\overline{\mu}]}$, $\delta$ is G-normally distributed and denoted by $N(0,[\underline{\sigma}^2,\overline{\sigma}^2])$. Specifically, the following two results can be obtained from the above theorem.

The first result is the law of large numbers (LLN) under sublinear expectations\footnote{In some literature, this result is also referred to as the weak law of large numbers (wLLN) under sublinear expectations.}
\begin{equation}
\lim_{n\rightarrow \infty} \mathbb{E}\left[\varphi\left(\sum_{k=1}^{n}\frac{\xi_k}{n}\right)\right] = \mathbb{E}[\varphi(\eta)],
\end{equation}
which shows that the uncertainty of distributions of $\frac{\sum_{k=1}^{n}\xi_k}{n}$ is approximately equivalent to all probability distribution of random variables taking values in $[\underline{\mu},\overline{\mu}]$. The second result is the central limit theorem (CLT) under sublinear expectations
\begin{equation}
\lim_{n\rightarrow \infty} \mathbb{E}\left[\varphi\left(\sum_{k=1}^{n}\frac{\zeta_k}{\sqrt{n}}\right)\right] = \mathbb{E}[\varphi(\delta)],
\end{equation}
which shows that the uncertainty of distributions of $\sum_{k=1}^{n}\frac{\zeta_k}{\sqrt{n}}$ converge to the G-normal distribution.

The maximal distribution serves to describe the uncertainty associated with the mean of the random variable, providing a framework that reflects the potential range of values that the mean can take. This is particularly useful when the mean is not precisely determined, allowing for a more flexible representation of the underlying mean uncertainty. The G-normal distribution addresses the uncertainty related to the variance. It is instrumental in describing the behavior of sums of independent random variables under conditions of variance uncertainty. Therefore, characterizing random variables with both uncertain means and variances can be effectively achieved by employing a combination of the maximal distribution and the G-normal distribution.

\subsection{Detection in additive noise channel with deterministic probability model}\label{subsec2.2}

In this subsection,we briefly introduce the detection problem in additive noise channels in information theory. We first introduce the relevant results of AWGN channels, and then extend them to fading channels. For more details, one can see \cite{tse2005fundamentals}.

Let $a$ be a complex number, $\mathrm{Re}(a)$ and $\mathrm{Im}(a)$ denote the real part and the imaginary part of $a$ respectively, i.e., $a=\mathrm{Re}(a)+i\mathrm{Im}(a)$. The complex conjugate of $a$ is denoted as $(a)^*$. $Q(v)=P(V>v)$, $V\sim \mathcal{N}(0,1)$.

Consider a real additive noise channel with white Gaussian noise, i.e., AWGN channel, given by
\begin{equation}
Y=X+Z, \quad Z\sim \mathcal{N}(0,\sigma^2),
\end{equation}
where $X,Y,Z$ denote input, output, and noise random variable, respectively. For the binary input signal $X\in \{x_A,x_B\} \subset \mathbb{R}$, $x_A>x_B$, the signal-to-noise ratio (SNR) is defined as $\frac{(x_A-x_B)^2}{8\sigma^2}$ \cite{tse2005fundamentals}. The detection problem requires judging whether $x_A$ or $x_B$ was the transmitted signal, given the observation $Y=y$. Determining which detector performs better is a significant and crucial question. The optimal detector for a channel represents the most effective strategy to ensure the transmission of information with minimal errors. Let $X$ be $x_A$ or $x_B$ with equiprobability. Then, based on the Bayes' rule, the optimal detector can be simplified as the \textit{minimun distance} detector and can be interpreted as choosing the nearest neighboring transmit symbol, i.e., choosing $x_A$ if
\begin{equation}
|y-x_A|\leq |y-x_B|.
\end{equation}
The probability of making an error is
\begin{equation}
P_e=Q\left(\frac{|x_A-x_B|}{2\sigma}\right)=Q(\sqrt{2\mathrm{SNR}}).
\end{equation}

Now consider a standard Rayleigh fading channel, characterized as 
\begin{equation}
Y=HX+Z,
\end{equation}
where $X,Y,Z,H$ denote input, output, noise and channel fading random variable, respectively, $H \sim \mathcal{CN}(0,1)$ and $Z \sim \mathcal{CN}(0,2\sigma^2)$. Suppose the receiver knows the channel fading, then the detection of $X=x$ from $Y=y$ can be done exactly as in the case of an AWGN channel, by simply equating this fading channel to the following real channel
\begin{equation}
S = |H|X + W,
\end{equation}
where $S=\mathrm{Re}\left( \left(\frac{H}{|H|}\right)^*Y \right)$ and $W \sim \mathcal{N}(0,\sigma^2)$. Then for a given value of $H=h$, the error probability of detecting $X$ is
\begin{equation}
P_{e,h}=Q\left(\frac{|h(x_A-x_B)|}{2\sigma}\right)=Q(\sqrt{2|h|^2\mathrm{SNR}}).
\end{equation}
Therefore, the overall error probability is
\begin{equation}
P_e = E[Q(\sqrt{2|H|^2\mathrm{SNR}})] = \frac{1}{2}\left(1-\sqrt{\frac{\mathrm{SNR}}{1+\mathrm{SNR}}}\right).
\end{equation}

Based on the discussion above, under the assumption of white Gaussian noise, the error probabilities, whether in AWGN channel or Rayleigh fading channel, depend solely on the distance between the two transmitted symbols $x_A$ and $x_B$. Under the framework of classical information theory, the performance of detection in additive noise channels can be precisely derived into explicit expressions.

\section{Additive noise channels with uncertain probability models}\label{sec3}

Consider the additive noise channel (\ref{additive noise channel}). Due to $Z(t_2)$ is IID to $Z(t_1)$ when $t_2>t_1$, we can focus on one symbol time and drop the time index. With the fixed input and output alphabets of the channel, the main factor affecting the performance of this channel becomes the noise within the channel. Consequently, when addressing the modeling of the channel with uncertain probability models, the distribution of noise is also uncertain.

In classical probability theory, mean and variance are important parameters of random variables, providing information about the central tendency and dispersion of the random variable. Determining the mean and variance of a random variable can greatly limit its distribution. In the context of nonlinear expectation theory, there are four crucial parameters $\mathbb{E}[Z]$, $-\mathbb{E}[-Z]$, $\mathbb{E}[Z^2]$ and $-\mathbb{E}[-Z^2]$ for random variable $Z$ on sublinear expectation space $(\Omega,\mathcal{H},\mathbb{E})$, which characterize the mean and variance uncertainties  of $Z$ and directly affect the range of uncertain distribution families of random variables. Hence, we consider two cases and propose the corresponding required concepts. The first is when both the mean and variance of the noise are uncertain, in this situation, we utilize a combination of the G-normal distribution and the maximum distribution to model the noise. We provide the following definition. The second is a special case of the first, occuring when there is no uncertainty in the mean but the variance is uncertain; and here, we apply the G-normal distribution to characterize the noise.

\begin{definition}\label{definition:AGDN}
	The \textit{additive distributionally uncertain noise (ADUN) channel} under sublinear expectation space $(\Omega,\mathcal{H},\mathbb{E})$ is given by
	\begin{equation}\label{AGDN channel}
	Y = X + Z.
	\end{equation}
	Here the noise random variable $Z$ satisfies
	\begin{equation}
	Z=M+\delta,
	\end{equation}
	where $M$ follows a maximum distribution, $\delta$ follows a G-normal distribution. $\delta$ is indepndent of $M$ under sublinear expectation $\mathbb{E}$.
	
	Furthermore, the \textit{ADUN channel without mean uncertainty} under sublinear expectation space $(\Omega,\mathcal{H},\mathbb{E})$ is given by
	\begin{equation}\label{AGNN channel}
	Y=X+\delta.
	\end{equation}
\end{definition}

\begin{remark}
	We assume that $\delta$ is independent of $M$ under $\mathbb{E}$ as defined in Definition~\ref{definition:independent}, and we require the equality (\ref{eq:independent}) to hold for all bounded measurable functions $\phi$ (including indicator functions). This is a modeling assumption motivated by the fact that in many practical scenarios, the mechanisms causing mean uncertainty and variance uncertainty are of distinct physical origins, which makes the assumption plausible in practice. The use of measurable functions is the natural framework for evaluating the upper and lower probabilities that characterize the channel.
\end{remark}

\section{Detection with uncertain probability models}\label{sec4}

We present the main results for the detection problem in additive noise channels with uncertain probability models, relying on a sublinear expectation space $(\Omega,\mathcal{H},\mathbb{E})$ where the uncertainty of the probability measures corresponding to $\mathbb{E}$ is denoted by $\mathcal{P}_{\Theta}=\{P_{\theta}\}_{\theta \in \Theta}$.

\subsection{Detection in ADUN channel}\label{subsec4.1}

Let us discuss the real ADUN channel, which is given in Definition \ref{definition:AGDN}, denoted as $Y=X+Z$. Here, $Z$ is the noise random variable, independent of $X$ under sublinear expectation $\mathbb{E}$, and satisfies
\begin{equation}
Z=M+\delta,
\end{equation}
where $M \overset{d}{=} M_{[\underline{\mu},\overline{\mu}]}$, $\delta \overset{d}{=} N(0,[\underline{\sigma}^2,\overline{\sigma}^2])$. $\delta$ is indepndent of $M$ under sublinear expectation $\mathbb{E}$.

With regard to random variables $Z$, $M$ and $\delta$, the following parameters are typically used in this model.
\begin{equation*}
\overline{\mu}=\mathbb{E}[M],\quad \underline{\mu}=-\mathbb{E}[-M],
\end{equation*}
\begin{equation*}
\overline{\sigma}^2=\mathbb{E}[\delta^2],\quad \underline{\sigma}^2=-\mathbb{E}[-\delta^2].
\end{equation*}
Specifically, because $\delta$ has no mean-uncertainty (we claimed it in Remark \ref{Remark:G-normal} in Section \ref{subsec2.1}), there are $\mathbb{E}[Z]=\overline{\mu}$ and $-\mathbb{E}[-Z]=\underline{\mu}$.

We consider the binary input signal $X\in \{x_A,x_B\} \subset \mathbb{R}$, $x_A>x_B$, and define the lower bound of SNR ($\underline{\mathrm{SNR}}$) and the upper bound of SNR ($\overline{\mathrm{SNR}}$) as follow. Let $\boldsymbol{M}=(M,\overline{M})$, where $\overline{M}$ is an independent copy of $M$, and define
\begin{equation}
\underline{\mathrm{SNR}}(\boldsymbol{M})=\frac{(x_A+M-x_B-\overline{M})^2}{8\overline{\sigma}^2},
\end{equation}
\begin{equation}
\overline{\mathrm{SNR}}(\boldsymbol{M})=\frac{(x_A+M-x_B-\overline{M})^2}{8\underline{\sigma}^2}.
\end{equation}
Then, the lower bound of SNR is characterized as
\begin{equation}
\begin{split}
&\underline{\mathrm{SNR}}=\mathcal{E}[\underline{\mathrm{SNR}}(\boldsymbol{M})]\\
=&\left\{
\begin{array}{ll}
\frac{(x_A+\underline{\mu}-x_B-\overline{\mu})^2}{8\overline{\sigma}^2}  &\quad \overline{\mu}-\underline{\mu} < x_A-x_B  \\
0  &\quad \overline{\mu}-\underline{\mu} \geq x_A-x_B
\end{array}
\right. ,
\end{split}
\end{equation}
and the upper bound of SNR is characterized as
\begin{equation}
\overline{\mathrm{SNR}}=\mathbb{E}[\overline{\mathrm{SNR}}(\boldsymbol{M})]=\frac{(x_A+\overline{\mu}-x_B-\underline{\mu})^2}{8\underline{\sigma}^2}.
\end{equation}
The SNR is defined as
\begin{equation}\label{definition:SNR}
\mathrm{SNR}=\frac{\underline{\mathrm{SNR}}+\overline{\mathrm{SNR}}}{2}.
\end{equation}

\begin{remark}
	In wireless communication, the SNR is generally considered as the ratio of average received signal energy to noise energy, and energy is generally regarded as variance. Affected by the uncertainties in both the mean and variance, for the ADUN channel, we are unable to determine the range of variance uncertainty for the noise $Z$. Therefore, we are not capable of defining the SNR directly. However, since $Y=X+Z$ is equivalent to $Y=X+M+\delta$, then the random variable $M$ can be seen as part of the transmitted signal, we can redefine the channel noise as $\delta$ and thus define the SNR. Note that treating $M$ as part of the transmitted signal is merely an idealized theoretical approach.
\end{remark}

For a given $X=x$, affected by the inherent uncertainty of the probability measures $\mathcal{P}_{\Theta}=\{P_{\theta}\}_{\theta \in \Theta}$, the performance of the channel is closely related to the family of probability distributions $\{P_{\theta}(Y\leq y)|P_{\theta}\in\mathcal{P}_{\Theta}\}$. To handle the uncertainty of probability distributions, the most ideal approach typically involves calculate all members in the set of possible probability distributions. However, since this family of probability measures is entirely unknown and the only thing we know is its relationship with sublinear expectation, not all probability distributions from this set can be determined. Nevertheless, based on the properties of maximum distribution and G-normal distribution, the supremum and infimum values of the set, i.e., $\underset{P_{\theta}\in\mathcal{P}_{\Theta}}{\sup}P_{\theta}(Y\leq y)$ and $\underset{P_{\theta}\in\mathcal{P}_{\Theta}}{\inf} P_{\theta}(Y\leq y)$ have explicit expressions, as demonstrated by the following theorem.

\begin{theorem}\label{theorem:AGDN1}
	For the ADUN channel $Y=X+Z$, where $Z=M+\delta$, $M \overset{d}{=} M_{[\underline{\mu},\overline{\mu}]}$, $\delta \overset{d}{=} N(0,[\underline{\sigma}^2,\overline{\sigma}^2])$, $\delta$ is indepndent of $M$, and $X=x$, the supremum and infimum values of the uncertain family of probabilities of $\{Y \leq y \}$ are given by $\overline{P}(y)$ and $\underline{P}(y)$, respectively. The explicit expressions are as follows:
	\begin{equation}\label{eq:AGNNmean upper}
	\overline{P}(y)=\left\{
	\begin{array}{lll}
	\frac{2\overline{\sigma}}{\overline{\sigma}+\underline{\sigma}}Q(\frac{\underline{\mu }+x-y}{\overline{\sigma}})  &\quad y\leq \underline{\mu}+x  \\
	\quad \\
	\frac{\overline{\sigma}-\underline{\sigma}}{\overline{\sigma}+\underline{\sigma}}+\frac{2\underline{\sigma}}{\overline{\sigma}+\underline{\sigma}}Q(\frac{\underline{\mu}+x-y}{\underline{\sigma}})  &\quad y>\underline{\mu}+x
	\end{array}
	\right. ,
	\end{equation}
	\begin{equation}\label{eq:AGNNmean lower}
	\underline{P}(y)=\left\{
	\begin{array}{lll}
	\frac{2\underline{\sigma}}{\overline{\sigma}+\underline{\sigma}}-\frac{2\underline{\sigma}}{\overline{\sigma}+\underline{\sigma}}Q(\frac{y-x-\overline{\mu }}{\underline{\sigma}})  &\quad y\leq x+\overline{\mu}  \\
	\quad \\
	1-\frac{2\overline{\sigma}}{\overline{\sigma}+\underline{\sigma}}Q(\frac{y-x-\overline{\mu}}{\overline{\sigma}})  &\quad y>x+\overline{\mu}
	\end{array}
	\right. .
	\end{equation}
\end{theorem}

{\bf Proof}\ \ Denote by $\mathcal{P}_{\Theta}=\{P_{\theta}\}_{\theta \in \Theta}$ the uncertain family of probability measures, and $\mathbb{E}[\cdot]:=\underset{P_{\theta}\in\mathcal{P}_{\Theta}}{\sup}E_{P_{\theta}}[\cdot]$ is the sublinear expectation, $\mathcal{E}[\cdot]=\underset{P_{\theta}\in\mathcal{P}_{\Theta}}{\inf}E_{P_{\theta}}[\cdot]$ is the conjugate expectation of $\mathbb{E}$. For a given $X=x$, the channel output $Y$ is equal to $x+M+\delta$. Therefore, we have
\begin{equation}
\overline{P}(y)=\underset{P_{\theta}\in\mathcal{P}_{\Theta}}{\sup}P_{\theta}(Y\leq y)=\mathbb{E}[I_{\{x+M+\delta \leq y \}}],
\end{equation}
\begin{equation}
\underline{P}(y)=\underset{P_{\theta}\in\mathcal{P}_{\Theta}}{\inf} P_{\theta}(Y\leq y)=\mathcal{E}[I_{\{x+M+\delta \leq y \}}].
\end{equation}

Since $\delta$ is independent of $M$, by Definition \ref{definition:independent}, there is
\begin{equation}
\mathbb{E}[I_{\{M+\delta\leq y-x\}}]=\mathbb{E}[\mathbb{E}[I_{\{m+\delta\leq y-x\}}]|m=M].
\end{equation}
Then, based on Definition \ref{maximal distribution} and Corollary 2 in \cite{peng2020hypothesis}, we obtain
\begin{equation}
\begin{split}
\mathbb{E}[I_{\{M+\delta\leq y-x\}}]=&\mathbb{E}[\frac{\sqrt{2}}{\overline{\sigma}+\underline{\sigma}}\int_{M+x-y}^{\infty}\frac{1}{\sqrt{\pi}}\left(e^{-\frac{z^2}{2\overline{\sigma}^2}}I_{\{z\geq 0\}} + e^{-\frac{z^2}{2\underline{\sigma}^2}}I_{\{z< 0\}}\right) dz]\\
=&\frac{\sqrt{2}}{\overline{\sigma}+\underline{\sigma}}\int_{\underline{\mu}+x-y}^{\infty}\frac{1}{\sqrt{\pi}}\left(e^{-\frac{z^2}{2\overline{\sigma}^2}}I_{\{z\geq 0\}} + e^{-\frac{z^2}{2\underline{\sigma}^2}}I_{\{z < 0\}}\right) dz\\
=&\left\{
\begin{array}{lll}
\frac{2\overline{\sigma}}{\overline{\sigma}+\underline{\sigma}}Q(\frac{\underline{\mu }+x-y}{\overline{\sigma}})  &\quad y\leq \underline{\mu}+x  \\
\quad \\
\frac{\overline{\sigma}-\underline{\sigma}}{\overline{\sigma}+\underline{\sigma}}+\frac{2\underline{\sigma}}{\overline{\sigma}+\underline{\sigma}}Q(\frac{\underline{\mu}+x-y}{\underline{\sigma}})  &\quad y>\underline{\mu}+x
\end{array}
\right. .
\end{split}
\end{equation}
Here we obtain the expression of $\overline{P}(y)$.

Similarly, we can get
\begin{equation}
\begin{split}
\mathbb{E}[I_{\{M+\delta> y-x\}}]=&\mathbb{E}[\frac{\sqrt{2}}{\overline{\sigma}+\underline{\sigma}}\int_{y-x-M}^{\infty}\frac{1}{\sqrt{\pi}}\left( e^{-\frac{z^2}{2\overline{\sigma}^2}}I_{\{z\geq 0\}} + e^{-\frac{z^2}{2\underline{\sigma}^2}}I_{\{z < 0\}} \right) dz]\\
=&\frac{\sqrt{2}}{\overline{\sigma}+\underline{\sigma}}\int_{y-x-\overline{\mu}}^{\infty}\frac{1}{\sqrt{\pi}}\left(e^{-\frac{z^2}{2\overline{\sigma}^2}}I_{\{z \geq 0\}} + e^{-\frac{z^2}{2\underline{\sigma}^2}}I_{\{z < 0\}} \right) dz\\
=&\left\{
\begin{array}{lll}
\frac{\overline{\sigma}-\underline{\sigma}}{\overline{\sigma}+\underline{\sigma}}+\frac{2\underline{\sigma}}{\overline{\sigma}+\underline{\sigma}}Q(\frac{y-x-\overline{\mu}}{\underline{\sigma}})  &\quad y \leq x+\overline{\mu}  \\
\quad \\
\frac{2\overline{\sigma}}{\overline{\sigma}+\underline{\sigma}}Q(\frac{y-x-\overline{\mu}}{\overline{\sigma}})  &\quad y > x+\overline{\mu}
\end{array}
\right. ,
\end{split}
\end{equation}
Then, according to $\mathcal{E}[I_{\{x+M+\delta \leq y \}}]=1-\mathbb{E}[I_{\{M+\delta> y-x\}}]$, we can get the expression of $\underline{P}(y)$. This conclude the proof. $\hfill\blacksquare$

\begin{corollary}\label{corollary:AGDN1}
	For the ADUN channel $Y=X+Z$, where $Z=M+\delta$, $M \overset{d}{=} M_{[\underline{\mu},\overline{\mu}]}$, $\delta \overset{d}{=} N(0,[\underline{\sigma}^2,\overline{\sigma}^2])$, $\delta$ is indepndent of $M$, and $X=x$, the supremum and infimum values of the uncertain family of probabilities of $\{Y > y \}$ are given by $\overline{T}(y)$ and $\underline{T}(y)$, respectively. The explicit expressions are as follows:
	\begin{equation}\label{eq:AGNNmean tail upper}
	\overline{T}(y)=\left\{
	\begin{array}{lll}
	\frac{\overline{\sigma}-\underline{\sigma}}{\overline{\sigma}+\underline{\sigma}}+\frac{2\underline{\sigma}}{\overline{\sigma}+\underline{\sigma}}Q(\frac{y-x-\overline{\mu}}{\underline{\sigma}})  &\quad y \leq x+\overline{\mu}  \\
	\quad \\
	\frac{2\overline{\sigma}}{\overline{\sigma}+\underline{\sigma}}Q(\frac{y-x-\overline{\mu}}{\overline{\sigma}})  &\quad y > x+\overline{\mu}
	\end{array}
	\right. ,
	\end{equation}
	\begin{equation}\label{eq:AGNNmean tail lower}
	\underline{T}(y)=\left\{
	\begin{array}{lll}
	1-\frac{2\overline{\sigma}}{\overline{\sigma}+\underline{\sigma}}Q(\frac{\underline{\mu}+x-y}{\overline{\sigma}})  &\quad y\leq \underline{\mu}+x  \\
	\quad \\
	\frac{2\underline{\sigma}}{\overline{\sigma}+\underline{\sigma}}-\frac{2\underline{\sigma}}{\overline{\sigma}+\underline{\sigma}}Q(\frac{\underline{\mu}+x-y}{\underline{\sigma}})  &\quad y>\underline{\mu}+x
	\end{array}
	\right. .
	\end{equation}
\end{corollary}

{\bf Proof}\ \  There is
\begin{equation}
\overline{T}(y)=\sup_{P_{\theta}\in\mathcal{P}_{\Theta}} P_{\theta}(Y > y)=\mathbb{E}[I_{\{M+\delta+x > y\}}],
\end{equation}
\begin{equation}
\underline{T}(y)=\inf_{P_{\theta}\in\mathcal{P}_{\Theta}} P_{\theta}(Y > y)=\mathcal{E}[I_{\{M+\delta+x > y\}}].
\end{equation}
Similar to the proof of Theorem \ref{theorem:AGDN1}, we can draw the conclusions. $\hfill\blacksquare$

\begin{remark}
	If $\overline{\mu}=\underline{\mu}=\mu$, $\overline{\sigma}=\underline{\sigma}=\sigma$, the conclusions in Theorem \ref{theorem:AGDN1} and Corollary \ref{corollary:AGDN1} degenerate into the traditional result about AWGN channel.
\end{remark}

\begin{definition}
	For a given family of decision rules, the \textit{optimal detector} is the rule that simultaneously minimizes both the maximum and minimum probabilities of error occurrence.
\end{definition}

The ``maximum probability'' and ``minimum probability'' refer to the maximum and minimum values among all possible probability values for the occurrence of error events, and can be intuitively understood as corresponding to conservative and aggressive strategies, respectively.

As we mentioned in Section \ref{subsec2.2}, for the real AWGN channel in classical information theory, if the two symbols $x_A,x_B$ are equally likely to have been transmitted, the optimal detector, with the least probability of making an erroneous decision, is known as minimum distance detector, and its analyses relies on the density function of normal probability distribution. In the context of nonlinear expectation theory, for ADUN channel, we still consider the transmitted signal $X$, which has an equal chance of being either $x_A$ or $x_B$, $x_A>x_B$. Since the transmitted signal is binary, we focus our attention on the family of threshold decision rules of the form: choose $x_A$ when the received signal $Y > y_0$ and $x_B$ otherwise, where $y_0$ is a threshold value to be determined. The optimal detector within this family is referred to as the optimal threshold detector. Then we have the following theorem, which illustrates a strong correlation between the optimal detector and the range of uncertainty in the noise mean. This theorem specifies two critical results: firstly, the optimal threshold detector exists if and only if the condition $\overline{\mu}-\underline{\mu} < x_A-x_B$ is satisfied; and secondly, the optimal threshold detector that exists under this condition is characterized as choosing $x_A$ when the received signal $Y > \frac{x_A+x_B+\overline{\mu}+\underline{\mu}}{2}$ and $x_B$ otherwise.

\begin{theorem}\label{theorem:AGNNmean detector}
	Consider the ADUN channel $Y=X+Z$, where $Z=M+\delta$, $M \overset{d}{=} M_{[\underline{\mu},\overline{\mu}]}$, $\delta \overset{d}{=} N(0,[\underline{\sigma}^2,\overline{\sigma}^2])$, $\delta$ is indepndent of $M$, and the channel input signal $X$ is equally likely to take on the values $x_A$ or $x_B$, $x_A>x_B$. Then, the optimal threshold detector exists just in case $\overline{\mu}-\underline{\mu} < x_A-x_B$. Under this case, the optimal threshold detector is choosing $x_A$ when the received signal $Y > \frac{x_A+x_B+\overline{\mu}+\underline{\mu}}{2}$ and $x_B$ otherwise, and the corresponding maximum probability (denoted as $\overline{P_e}$) and minimum probability (denoted as $\underline{P_e}$) of error occurrence are
	\begin{equation}\label{eq:sup expression}
	\overline{P_e} = \frac{2\overline{\sigma}}{\overline{\sigma}+\underline{\sigma}}Q(\frac{x_A+\underline{\mu}-x_B-\overline{\mu}}{2\overline{\sigma}}),
	\end{equation}
	\begin{equation}\label{eq:inf expression}
	\underline{P_e} = \frac{2\underline{\sigma}}{\overline{\sigma}+\underline{\sigma}}Q(\frac{x_A+\overline{\mu}-x_B-\underline{\mu}}{2\underline{\sigma}}).
	\end{equation}
\end{theorem}

{\bf Proof}\ \  Denote by $\mathcal{P}_{\Theta}=\{P_{\theta}\}_{\theta \in \Theta}$ the uncertain family of probability measures, and $\mathbb{E}[\cdot]:=\underset{P_{\theta}\in\mathcal{P}_{\Theta}}{\sup}E_{P_{\theta}}[\cdot]$ is the sublinear expectation. Suppose the detector is choosing $x_A$ for $Y > y_0$, and $x_B$ for $Y \leq y_0$. We denote the error event as $\{\mathrm{error|y_0} \}$. Then,
\begin{equation}
\overline{P_e}(y_0):=\sup_{P_{\theta}\in\mathcal{P}_{\Theta}} P_{\theta}(\{error|y_0\})=\mathbb{E}[I_{ \{error|y_0\} }],
\end{equation}
\begin{equation}
\underline{P_e}(y_0):=\inf_{P_{\theta}\in\mathcal{P}_{\Theta}} P_{\theta}(\{error|y_0\})=\mathcal{E}[I_{ \{error|y_0\} }].
\end{equation}

Because $X$ is equal to $x_A$ or $x_B$ equiprobability, $Z$ is independent of $X$ under sublinear expectation $\mathbb{E}$, we have
\begin{equation}
\mathbb{E}[I_{ \{error|y_0\} }]	=\frac{1}{2}\sup_{P_{\theta}\in\mathcal{P}_{\Theta}} P_{\theta}(Y \leq y_0|X=x_A) + \frac{1}{2}\sup_{P_{\theta}\in\mathcal{P}_{\Theta}} P_{\theta}(Y > y_0|X=x_B),
\end{equation}
\begin{equation}
\mathcal{E}[I_{ \{error|y_0\} }] = \frac{1}{2}\inf_{P_{\theta}\in\mathcal{P}_{\Theta}} P_{\theta}(Y \leq y_0|X=x_A) + \frac{1}{2}\inf_{P_{\theta}\in\mathcal{P}_{\Theta}} P_{\theta}(Y > y_0|X=x_B).
\end{equation}

Case 1. If $\overline{\mu}-\underline{\mu} < x_A-x_B$, by combining eqs. (\ref{eq:AGNNmean upper}), (\ref{eq:AGNNmean lower}), (\ref{eq:AGNNmean tail upper}) and (\ref{eq:AGNNmean tail lower}), we can obtain
\begin{equation}\label{eq:sup}
\begin{split}
&\overline{P_e}(y_0):= \sup_{P_{\theta}\in\mathcal{P}_{\Theta}} P_{\theta}(\{error|y_0\}) \\
=&\left\{
\begin{array}{lllll}
\frac{\overline{\sigma}-\underline{\sigma}}{2(\overline{\sigma}+\underline{\sigma})}+\frac{\underline{\sigma}}{\overline{\sigma}+\underline{\sigma}}Q(\frac{\underline{\mu}+x_A-y_0}{\underline{\sigma}})+\frac{\overline{\sigma}}{\overline{\sigma}+\underline{\sigma}}Q(\frac{y_0-x_B-\overline{\mu}}{\overline{\sigma}})  &\quad y_0>x_A+\underline{\mu}  \\
\quad \\
\frac{\overline{\sigma}}{\overline{\sigma}+\underline{\sigma}}Q(\frac{\underline{\mu}+x_A-y_0}{\overline{\sigma}})+\frac{\overline{\sigma}}{\overline{\sigma}+\underline{\sigma}}Q(\frac{y_0-x_B-\overline{\mu}}{\overline{\sigma}})  &\quad x_B+\overline{\mu} < y_0 \leq x_A+\underline{\mu}\\
\quad \\
\frac{\overline{\sigma}-\underline{\sigma}}{2(\overline{\sigma}+\underline{\sigma})}+\frac{\overline{\sigma}}{\overline{\sigma}+\underline{\sigma}}Q(\frac{\underline{\mu}+x_A-y_0}{\overline{\sigma}})+\frac{\underline{\sigma}}{\overline{\sigma}+\underline{\sigma}}Q(\frac{y_0-x_B-\overline{\mu}}{\underline{\sigma}})  &\quad y_0 \leq x_B+\overline{\mu} 
\end{array}
\right. ,
\end{split}
\end{equation}
and
\begin{equation}\label{eq:inf}
\begin{split}
&\underline{P_e}(y_0):=\inf_{P_{\theta}\in\mathcal{P}_{\Theta}} P_{\theta}(\{error|y_0\}) \\
=&\left\{
\begin{array}{lllll}
\frac{1}{2}+\frac{\underline{\sigma}}{\overline{\sigma}+\underline{\sigma}}-\frac{\overline{\sigma}}{\overline{\sigma}+\underline{\sigma}}Q(\frac{y_0-x_A-\overline{\mu}}{\overline{\sigma}})-\frac{\underline{\sigma}}{\overline{\sigma}+\underline{\sigma}}Q(\frac{\underline{\mu}+x_B-y_0}{\underline{\sigma}})  &\quad y_0>x_A+\underline{\mu}  \\
\quad \\
\frac{2\underline{\sigma}}{\overline{\sigma}+\underline{\sigma}}-\frac{\underline{\sigma}}{\overline{\sigma}+\underline{\sigma}}Q(\frac{y_0-x_A-\overline{\mu}}{\underline{\sigma}})-\frac{\underline{\sigma}}{\overline{\sigma}+\underline{\sigma}}Q(\frac{\underline{\mu}+x_B-y_0}{\underline{\sigma}})  &\quad x_B+\overline{\mu} < y_0 \leq x_A+\underline{\mu} \\
\quad \\
\frac{1}{2}+\frac{\underline{\sigma}}{\overline{\sigma}+\underline{\sigma}}-\frac{\underline{\sigma}}{\overline{\sigma}+\underline{\sigma}}Q(\frac{y_0-x_A-\overline{\mu}}{\underline{\sigma}})-\frac{\overline{\sigma}}{\overline{\sigma}+\underline{\sigma}}Q(\frac{\underline{\mu}+x_B-y_0}{\overline{\sigma}})  &\quad y_0 \leq x_B+\overline{\mu} 
\end{array}
\right. .
\end{split}
\end{equation}

Taking minimum to $\overline{P_e}(y_0)$ and $\underline{P_e}(y_0)$ both yield the result that $y_0$ equals $\frac{x_A+x_B+\overline{\mu}+\underline{\mu}}{2}$. By substituting $y_0 = \frac{x_A+x_B+\overline{\mu}+\underline{\mu}}{2}$ into eqs. (\ref{eq:sup}) and (\ref{eq:inf}), we can derive the expressions for $\overline{P_e}$ and $\underline{P_e}$ as (\ref{eq:sup expression}) and (\ref{eq:inf expression}), respectively.

Case 2. If $\overline{\mu}-\underline{\mu}>x_A-x_B$, $\underline{P_e}(y_0)$ is same as equation (\ref{eq:inf}). Taking minimum to $\underline{P_e}$ lead to $y_0=\frac{x_A+x_B+\overline{\mu}+\underline{\mu}}{2}$. However, there is
\begin{equation}
\overline{P_e}(y_0)=\left\{
\begin{array}{lllll}
\frac{\overline{\sigma}-\underline{\sigma}}{2(\overline{\sigma}+\underline{\sigma})}+\frac{\underline{\sigma}}{\overline{\sigma}+\underline{\sigma}}Q(\frac{\underline{\mu}+x_A-y_0}{\underline{\sigma}})+\frac{\overline{\sigma}}{\overline{\sigma}+\underline{\sigma}}Q(\frac{y_0-x_B-\overline{\mu}}{\overline{\sigma}})  &\quad y_0>x_B+\overline{\mu}  \\
\quad \\
\frac{\overline{\sigma}-\underline{\sigma}}{\overline{\sigma}+\underline{\sigma}}+\frac{\underline{\sigma}}{\overline{\sigma}+\underline{\sigma}}Q(\frac{\underline{\mu}+x_A-y_0}{\underline{\sigma}})+\frac{\underline{\sigma}}{\overline{\sigma}+\underline{\sigma}}Q(\frac{y_0-x_B-\overline{\mu}}{\underline{\sigma}})  &\quad x_A+\underline{\mu} < y_0 \leq x_B+\overline{\mu}\\
\quad \\
\frac{\overline{\sigma}-\underline{\sigma}}{2(\overline{\sigma}+\underline{\sigma})}+\frac{\overline{\sigma}}{\overline{\sigma}+\underline{\sigma}}Q(\frac{\underline{\mu}+x_A-y_0}{\overline{\sigma}})+\frac{\underline{\sigma}}{\overline{\sigma}+\underline{\sigma}}Q(\frac{y_0-x_B-\overline{\mu}}{\underline{\sigma}})  &\quad y_0 \leq x_A+\underline{\mu} 
\end{array}
\right. .
\end{equation}
We can calculate that $\overline{P_e}(y_0)$ is always greater than $\frac{1}{2}$ and the minimum of $\overline{P_e}(y_0)$ is $\frac{1}{2}$ when $y_0$ tends to $-\infty$ or $\infty$. Therefore, the optimal threshold detector does not exist within the parameters of this situation.

This conclude the proof of Theorem \ref{theorem:AGNNmean detector}. $\hfill\blacksquare$

In our analyses, we have proved that the condition $\overline{\mu}-\underline{\mu} < x_A-x_B$ is directly correlated with the existence of an optimal detector, which is crucial for the effective operation of the detection mechanism. Following the estimation of channel parameters using pilot \cite{dong2004optimal}, it becomes necessary to adjust the values of $x_A$ and $x_B$ to ensure that this condition is satisfied. Given that the channel noise is independent of the input signals, these adjustments are feasible. By modifying $x_A$ and $x_B$, we can configure the system parameters to meet the required inequality, thereby enabling the deployment of an optimal detector that enhances the accuracy and reliability of signal detection within the communication channel. If the condition is violated, i.e., $\overline{\mu}-\underline{\mu} \ge x_A-x_B$, the two constellation points are not sufficiently separated relative to the uncertainty range of the noise mean. In this case, the upper and lower error probabilities can no longer be jointly minimized by any single decision rule; in particular, the lower error probability remains appreciable even in the best-case noise realization, and any detector must always tolerate a non-negligible worst-case error floor. The significance of this condition is also reflected in the definition of SNR. If the inequality $\overline{\mu}-\underline{\mu} < x_A-x_B$ is not satisfied, the $\underline{\mathrm{SNR}}$ becomes invalid.

Furthermore, when the transmit power is strictly limited and the existence of an optimal threshold detector cannot be maintained simply by increasing the spacing between the constellation points, reliable communication under such power constraints can also employ mechanisms that do not depend solely on the distance between individual symbols. For example, diversity techniques can be adopted, where replicas of the same information are transmitted over multiple independent channels (e.g., different time slots, frequencies, or antennas) to disperse the impact of noise uncertainty. The essence of this approach is to base the decision on a vector observation composed of multiple symbols, accumulating the distinguishability between signal and noise by increasing the dimensionality of the signal space. In this way, the overall reliability of the system can be sustained even when the per symbol energy can no longer be increased.

In Theorem \ref{theorem:AGNNmean detector}, we derive that the deterministic boundary for the problem of detecting $X$ is given by $\frac{x_A+x_B+\overline{\mu}+\underline{\mu}}{2}$. This finding diverges from the conclusions drawn in classical information theory, as it illustrates the optimal detector is affected by the mean uncertainty of the channel noise and the range of this uncertainty. This conclusion holds significant relevance for processing real-world data, as it accounts for the fact that even a set of samples has an arithmetic mean of zero, the upper mean $\overline{\mu}$ and the lower mean $\underline{\mu}$ of the sample set can differ due to the inherent uncertainty of the probability model itself. This discrepancy is crucial for accurately modeling and interpreting data in practical applications, as it reflects the impact of mean uncertainty on the performance of the optimal detector.

\begin{remark}
	When $\overline{\sigma}=\underline{\sigma}=\sigma$ and $\overline{\mu}=\underline{\mu}=0$, the conclusions in Theorem \ref{theorem:AGNNmean detector} degenerate into the traditional results for that in AWGN channel.
\end{remark}

For ADUN channel, the $\varphi$-max-mean algorithm \cite{peng2017theory} within the nonlinear expectation theory can be employed to estimate the range of mean uncertainty in the noise derectly. However, there is no suitable method available for estimating the range of variance uncertainty in the noise. Consequently, we present below a method for estimating the variance of noise that accounts for both mean and variance uncertainties.

Based on the previous discussion, $\underset{P_{\theta}\in\mathcal{P}_{\Theta}}{\sup} P_{\theta}(Y\leq y_0|X)$ and $\underset{P_{\theta}\in\mathcal{P}_{\Theta}}{\inf} P_{\theta}(Y\leq y_0|X)$ can be considered as function of $\overline{\mu}$, $\underline{\mu}$, $\overline{\sigma}$ and $\underline{\sigma}$. Similar to \cite{ji2023imbalanced}, there hold
\begin{equation}\label{estimate_1}
\mathbb{E}[I_{\{Y\leq y_0\}}-\sup_{P_{\theta}\in\mathcal{P}_{\Theta}} P_{\theta}(Y\leq y_0|X)]=0,
\end{equation}
\begin{equation}\label{estimate_2}
\mathcal{E}[I_{\{Y\leq y_0\}}-\inf_{P_{\theta}\in\mathcal{P}_{\Theta}} P_{\theta}(Y\leq y_0|X)]=0.
\end{equation}
Let $n$ be sample size and $m\leq n$ be constant. By Theorem 24 in \cite{jin2021optimal}, we obtain that the terms
\begin{equation}\label{estimate_3}
\max_{0\leq k \leq n-m}\frac{1}{m}\sum_{l=1}^{m}[I_{\{y_{k+l}\leq y_0\}}-\frac{\sqrt{2}}{\overline{\sigma}+\underline{\sigma}}\int_{\underline{\mu}+x_{k+l}-y_0}^{\infty}\frac{1}{\sqrt{\pi}}\left(e^{-\frac{z^2}{2\overline{\sigma}^2}}I_{\{z\geq 0\}} + e^{-\frac{z^2}{2\underline{\sigma}^2}}I_{\{z< 0\}}\right) dz]
\end{equation}
and
\begin{equation}\label{estimate_4}
\min_{0\leq k \leq n-m}\frac{1}{m}\sum_{l=1}^{m}[I_{\{y_{k+l}\leq y_0\}}-1+\frac{\sqrt{2}}{\overline{\sigma}+\underline{\sigma}}\int_{y_0-x_{k+l}-\overline{\mu}}^{\infty}\frac{1}{\sqrt{\pi}}\left(e^{-\frac{z^2}{2\overline{\sigma}^2}}I_{\{z\geq 0\}} + e^{-\frac{z^2}{2\underline{\sigma}^2}}I_{\{z< 0\}}\right) dz]
\end{equation}
converge to the left sides of \eqref{estimate_1} and \eqref{estimate_2} as $n\rightarrow \infty$, respectively. Then, we can summary the following Algorithm \ref{alg:mean-variance uncertainty} to estimate the parameters $\overline{\mu}$, $\underline{\mu}$, $\overline{\sigma}$ and $\underline{\sigma}$, and the performance of detection in ADUN channel, after providing a set of channel input and output data.

\begin{algorithm}[!h]
	\caption{Parameters and detection performance estimation method of additive noise channel with mean and variance uncertainty}
	\label{alg:mean-variance uncertainty}
	\renewcommand{\algorithmicrequire}{\textit{Input:}}
	\renewcommand{\algorithmicensure}{\textit{Output:}}
	\begin{algorithmic}[1]
		\REQUIRE Sample set $\{x_i\}_{i=1}^{n}$, $x_i\in \{x_A,x_B\}$. The sliding window size $m\leq n$.
		\ENSURE $\hat{\overline{\mu}}$, $\hat{\underline{\mu}}$, $\hat{\overline{\sigma}}$ and $\hat{\underline{\sigma}}$. $\hat{\overline{P_e}}$ and $\hat{\underline{P_e}}$.
		
		\STATE $\{x_i\}_{i=1}^{n}$ through the channel $Y=X+Z$ obtain $\{y_i\}_{i=1}^{n}$, where $Z=M+\delta$ is a noise with mean uncertainty and variance uncertainty, i.e., $M \overset{d}{=} M_{[\underline{\mu},\overline{\mu}]}$, $\delta \overset{d}{=} N(0,[\underline{\sigma}^2,\overline{\sigma}^2])$.
		
		\STATE Calculate $\hat{\overline{\mu}}$ and $\hat{\underline{\mu}}$ as the estimate of $\overline{\mu}$ and $\underline{\mu}$, and $\hat{y}_0=\frac{x_A+x_B+\hat{\overline{\mu}}+\hat{\underline{\mu}}}{2}$, where
		\begin{equation*}
		\hat{\overline{\mu}} = \max_{0\leq k \leq n-m}\frac{1}{m}\sum_{l=1}^{m}(y_l-x_l), \ \ \hat{\underline{\mu}} = \min_{0\leq k \leq n-m}\frac{1}{m}\sum_{l=1}^{m}(y_l-x_l).
		\end{equation*}
		
		\STATE Due to \eqref{estimate_1} and \eqref{estimate_2}, we obtain $\hat{\overline{\sigma}}$ and $\hat{\underline{\sigma}}$ as the estimate of $\overline{\sigma}$ and $\underline{\sigma}$ by solving the numerical result of equations $(\ref{estimate_3})=0$ and $(\ref{estimate_4})=0$, in which the $\overline{\mu}$, $\underline{\mu}$ and $y_0$ in the terms (\ref{estimate_3}) and (\ref{estimate_4}) are substituted with $\hat{\overline{\mu}}$, $\hat{\underline{\mu}}$ and $\hat{y}_0$.
		
		\STATE Approximate $\overline{P_e}$ and $\underline{P_e}$ as
		\begin{equation*}
		\hat{\overline{P_e}}=\frac{2\hat{\overline{\sigma}}}{\hat{\overline{\sigma}}+\hat{\underline{\sigma}}}Q(\frac{x_A+\hat{\underline{\mu}}-x_B-\hat{\overline{\mu}}}{2\hat{\overline{\sigma}}}), \ \ \hat{\underline{P_e}}=\frac{2\hat{\underline{\sigma}}}{\hat{\overline{\sigma}}+\hat{\underline{\sigma}}}Q(\frac{x_A+\hat{\overline{\mu}}-x_B-\hat{\underline{\mu}}}{2\hat{\underline{\sigma}}}).
		\end{equation*}
	\end{algorithmic}
\end{algorithm}

\begin{remark}\label{remark:variance}
	In ADUN channel without mean uncertainty, the range of variance uncertainty can be directly estimated by the $\varphi$-max-mean algorithm within the nonlinear expectation theory, based on the properties $\overline{\sigma}^2=\mathbb{E}[\delta^2]$ and $\underline{\sigma}^2=-\mathbb{E}[-\delta^2]$.
\end{remark}

\begin{remark}
	The estimators $\hat{\overline{\mu}}$ and $\hat{\underline{\mu}}$ in Algorithm \ref{alg:mean-variance uncertainty} are essentially the sample extrema of moving-window averages via the $\varphi$-max-mean algorithm, which serve as natural estimators for the upper and lower means $\overline{\mu}$ and $\underline{\mu}$. A relevant theoretical benchmark for understanding the finite sample behavior of such estimators is provided by the convergence rate of the law of large numbers under sublinear expectations. Specifically, \cite{song2021stein} established that for IID sequences with $\mathbb{E}[|X_1|^2]<\infty$, the convergence rate is bounded by $\frac{C}{\sqrt{n}}$. These theoretical guarantees establish the consistency of Algorithm \ref{alg:mean-variance uncertainty} and provide a quantitative characterization of its finite-sample behavior. From a computational perspective, the main cost lies in solving the empirical equations $(\ref{estimate_3})=0$ and $(\ref{estimate_4})=0$. For a fixed sample size $n$ and sliding window size $m$, a direct evaluation of computational complexity costs $\mathcal{O}((n-m)m)$. By updating the sliding-window sums recursively, the per-iteration computational cost is reduced to $\mathcal{O}(n)$. In summary, these analyses confirm that Algorithm \ref{alg:mean-variance uncertainty} is both statistically consistent and computationally practical for typical pilot sequence lengths.
\end{remark}

\begin{remark}
	It should be noted that the estimation procedure described above assumes the presence of variance uncertainty, i.e., $\overline{\sigma} > \underline{\sigma}$. When $\overline{\sigma} = \underline{\sigma}$, the two estimating equations in Step~3 of Algorithm~\ref{alg:mean-variance uncertainty} may become linearly dependent, leading to a potential identifiability issue. In practice, one can use the estimates of $\overline{\mu}$ and $\underline{\mu}$ obtained in Step~2 to guide the subsequent variance estimation: if a non-negligible gap between $\hat{\overline{\mu}}$ and $\hat{\underline{\mu}}$ is observed, it is reasonable to proceed with Step~3 to estimate $\hat{\overline{\sigma}}$ and $\hat{\underline{\sigma}}$; if, on the other hand, the mean uncertainty is found to be insignificant, the variance can be directly estimated based on Remark \ref{remark:variance}, thereby avoiding the numerically unstable questions in Step~3. For the theoretical analysis in this paper, we focus on the non‑degenerate case $\overline{\sigma} > \underline{\sigma}$, which is the primary regime of interest when uncertainty is genuinely present.
\end{remark}

\subsection{Detection in ADUN channel without mean uncertainty}\label{subsec4.2}

Let us consider the real ADUN channel without mean uncertainty, given by equation (\ref{AGNN channel}). Suppose $\delta \overset{d}{=} N(0,[\underline{\sigma}^2,\overline{\sigma}^2])$ and $\delta$ is independent of $X$ under sublinear expectation $\mathbb{E}$. For the binary input signal $X\in \{x_A,x_B\} \subset \mathbb{R}$, $x_A>x_B$, we define the lower bound of SNR ($\underline{\mathrm{SNR}}$) and the upper bound of SNR ($\overline{\mathrm{SNR}}$) as
\begin{equation}
\underline{\mathrm{SNR}}=\frac{(x_A-x_B)^2}{8\overline{\sigma}^2}, \quad \overline{\mathrm{SNR}}=\frac{(x_A-x_B)^2}{8\underline{\sigma}^2}.
\end{equation}
The definition of SNR is the same as (\ref{definition:SNR}).

\begin{remark}
	In dB units, there is the following relationship
	\begin{equation}
	\mathrm{SNR(dB)}=\frac{\underline{\mathrm{SNR}}\mathrm{(dB)}+\overline{\mathrm{SNR}}\mathrm{(dB)}}{2}+10\lg\frac{\overline{\sigma}^2+\underline{\sigma}^2}{2\overline{\sigma}\underline{\sigma}}.
	\end{equation}
\end{remark}

\begin{remark}
	Using binary input signals, for ADUN channel without mean uncertainty, with $\mathbb{E}[\delta^2]=\overline{\sigma}^2$ and $-\mathbb{E}[-\delta^2]=\underline{\sigma}^2$, its SNR is equivalent to the SNR of the traditional AWGN channel with the noise following $\mathcal{N}(0,\sigma^2)$ where $\sigma^2=\frac{2\overline{\sigma}^2\underline{\sigma}^2}{\overline{\sigma}^2+\underline{\sigma}^2}$.
\end{remark}

Since this is a special case of the previous section, we present the following conclusions without proof in this subsection. For a given $X=x$, there is also a set of probability distributions $\{P_{\theta}(Y\leq y)|P_{\theta}\in\mathcal{P}_{\Theta}\}$, which is crucial for characterizing the performance of detection problem in ADUN channel without mean uncertainty. The following theorem determines the maximum value and the minimum value of this set, i.e., $\underset{P_{\theta}\in\mathcal{P}_{\Theta}}{\sup}P_{\theta}(Y\leq y)$ and $\underset{P_{\theta}\in\mathcal{P}_{\Theta}}{\inf} P_{\theta}(Y\leq y)$.

\begin{theorem}\label{theorem:AGNN1}
	For the ADUN channel without mean uncertainty $Y=X+\delta$ with $\delta \overset{d}{=} N(0,[\underline{\sigma}^2,\overline{\sigma}^2])$, and a given $X=x$, the supremum and infimum values of the uncertain family of probabilities of $\{Y \leq y \}$ are given by $\overline{P'}(y)$ and $\underline{P'}(y)$, respectively. The explicit expressions are as follows:
	\begin{equation}\label{eq:AGNN upper}
	\overline{P'}(y)=\left\{
	\begin{array}{lll}
	\frac{2\overline{\sigma}}{\overline{\sigma}+\underline{\sigma}}Q(\frac{x-y}{\overline{\sigma}})  &\quad y\leq x  \\
	\quad \\
	\frac{\overline{\sigma}-\underline{\sigma}}{\overline{\sigma}+\underline{\sigma}}+\frac{2\underline{\sigma}}{\overline{\sigma}+\underline{\sigma}}Q(\frac{x-y}{\underline{\sigma}})  &\quad y>x
	\end{array}
	\right. ,
	\end{equation}
	\begin{equation}\label{eq:AGNN lower}
	\underline{P'}(y)=\left\{
	\begin{array}{ll}
	\frac{2\underline{\sigma}}{\overline{\sigma}+\underline{\sigma}}-\frac{2\underline{\sigma}}{\overline{\sigma}+\underline{\sigma}}Q(\frac{y-x}{\underline{\sigma}})  &\quad y \leq x  \\
	\quad \\
	1-\frac{2\overline{\sigma}}{\overline{\sigma}+\underline{\sigma}}Q(\frac{y-x}{\overline{\sigma}})  &\quad y > x
	\end{array}
	\right. .
	\end{equation}
\end{theorem}

\begin{corollary}\label{corollary:AGNN1}
	For the ADUN channel without mean uncertainty $Y=X+\delta$ with $\delta \overset{d}{=} N(0,[\underline{\sigma}^2,\overline{\sigma}^2])$, and a given $X=x$, the supremum and infimum values of the uncertain family of probabilities of $\{Y > y \}$ are given by $\overline{T'}(y)$ and $\underline{T'}(y)$, respectively. The explicit expressions are as follows:
	\begin{equation}\label{eq:AGNN tail upper}
	\overline{T'}(y)=\left\{
	\begin{array}{lll}
	\frac{\overline{\sigma}-\underline{\sigma}}{\overline{\sigma}+\underline{\sigma}}+\frac{2\underline{\sigma}}{\overline{\sigma}+\underline{\sigma}}Q(\frac{y-x}{\underline{\sigma}})  &\quad y \leq x  \\
	\quad \\
	\frac{2\overline{\sigma}}{\overline{\sigma}+\underline{\sigma}}Q(\frac{y-x}{\overline{\sigma}})  &\quad y > x
	\end{array}
	\right. ,
	\end{equation}
	\begin{equation}\label{eq:AGNN tail lower}
	\underline{T'}(y)=\left\{
	\begin{array}{lll}
	1-\frac{2\overline{\sigma}}{\overline{\sigma}+\underline{\sigma}}Q(\frac{x-y}{\overline{\sigma}})  &\quad y\leq x  \\
	\quad \\
	\frac{2\underline{\sigma}}{\overline{\sigma}+\underline{\sigma}}-\frac{2\underline{\sigma}}{\overline{\sigma}+\underline{\sigma}}Q(\frac{x-y}{\underline{\sigma}})  &\quad y>x
	\end{array}
	\right. .
	\end{equation}
\end{corollary}

\begin{remark}
	If $\overline{\sigma}=\underline{\sigma}=\sigma$, the conclusions in Theorem \ref{theorem:AGNN1} and Corollary \ref{corollary:AGNN1} degenerate into the traditional result about AWGN channel.
\end{remark}

Regarding the ADUN channel without mean uncertainty that is focused in this subsection, the theorem presented below illustrates that the minimum distance rule is still the optimal threshold detector when the noise within the channel has no mean uncertainty but has variance uncertainty.

\begin{theorem}\label{theorem:AGNN detector}
	Consider the ADUN channel without mean uncertainty $Y=X+\delta$ with $\delta \overset{d}{=} N(0,[\underline{\sigma}^2,\overline{\sigma}^2])$, and the channel input signal $X$ is equally likely to take on the values $x_A$ or $x_B$, $x_A>x_B$. Then the minimum distance rule, choosing $x_A$ when the received signal $Y > \frac{x_A+x_B}{2}$ and $x_B$ otherwise, is the optimal threshold detector for this system, and the corresponding maximum probability (denoted as $\overline{P_e}$) and minimum probability (denoted as $\underline{P_e}$) of error occurrence are
	\begin{equation}
	\overline{P_e} = \frac{2\overline{\sigma}}{\overline{\sigma}+\underline{\sigma}}Q(\sqrt{2\underline{\mathrm{SNR}}}),
	\end{equation}
	\begin{equation}
	\underline{P_e} = \frac{2\underline{\sigma}}{\overline{\sigma}+\underline{\sigma}}Q(\sqrt{2\overline{\mathrm{SNR}}}).
	\end{equation}
\end{theorem}

\begin{remark}
	When $\overline{\sigma}=\underline{\sigma}=\sigma$, the conclusions in Theorem \ref{theorem:AGNN detector} degenerate into the tradition results for that in AWGN channel.
\end{remark}

\section{From additive noise channel to fading channel}\label{sec5}

In the preceding sections of this paper, we have focused our analyses on additive noise channels without fading, providing the definitions of ADUN channel and ADUN channel without mean uncertainty. This simplification allowed us to delve into the fundamental aspects of signal transmission and reception under ideal conditions. However, real-world communication channels are often subject to fading, a phenomenon that can significantly degrade signal quality. In this section, we transition from the idealized channels to a more realistic one by considering fading channels. Specifically, we will address the case where the receiver has knowledge of the channel fading, a condition often referred to as ``known channel state information at the receiver (CSIR)" \cite{goldsmith1997capacity}. Despite the presence of distribution uncertainty in the channel noise, we will demonstrate that the detection problem in fading channels can be effectively transformed into a problem analogous to that of additive noise channels. This transformation enables us to utilize the conclusions developed in ADUN channel and ADUN channel without mean uncertainty to tackle scenarios where the channel has fading.

Consider the channel with complex random variables of Rayleigh fading and noise
\begin{equation}
Y'=HX+Z',
\end{equation}
where $H \sim \mathcal{CN}(0,1)$ is the channel fading. For the binary input signal $X\in \{x_A,x_B\} \subset \mathbb{R}$, $x_A>x_B$, the detection problem can be equivalent to one based on the sign of the real sufficient statistics
\begin{equation}\label{eq:real sufficient statistics}
S'=|H|X+W',
\end{equation}
where $S'=\mathrm{Re}(\overline{H}Y')$, $W'=\mathrm{Re}(\overline{H}Z')$, $\overline{H}=(\frac{H}{|H|})^*$. If $W'$ is G-normally distributed, then, based on the theoretical framework established in Section \ref{subsec4.1}, the performance of the detection problem can be characterized. Given $H=h$, based on Theorem \ref{theorem:AGNN detector}, the maximum error probability and minimum error probability of detecting $X$ are
\begin{equation}
\overline{P_{e,h}}=\frac{2\overline{\sigma}}{\overline{\sigma}+\underline{\sigma}}Q\left(\frac{|h|(x_A-x_B)}{2\overline{\sigma}}\right),
\end{equation}
\begin{equation}
\underline{P_{e,h}}=\frac{2\underline{\sigma}}{\overline{\sigma}+\underline{\sigma}}Q\left(\frac{|h|(x_A-x_B)}{2\underline{\sigma}}\right).
\end{equation}
Then, we average over the random channel fading $H$, by directly integration, and yield
\begin{equation}
\overline{P_{e}}=E[\overline{P_{e,H}}]=\frac{\overline{\sigma}}{\overline{\sigma}+\underline{\sigma}}[1-\sqrt{\frac{(x_A-x_B)^2}{(x_A-x_B)^2+8\overline{\sigma}^2}}],
\end{equation}
\begin{equation}
\underline{P_{e}}=E[\underline{P_{e,H}}]=\frac{\underline{\sigma}}{\overline{\sigma}+\underline{\sigma}}[1-\sqrt{\frac{(x_A-x_B)^2}{(x_A-x_B)^2+8\underline{\sigma}^2}}].
\end{equation}

To ensure that $W'$ follows a G-normal distribution, the following theorem provides a sufficient condition.

\begin{theorem}\label{theorem:sufficient G-normal}
	Let $(\Omega,\mathcal{H},\mathbb{E})$ be the sublinear expectation space. $Z'=\mathrm{Re}(Z')+i\mathrm{Im}(Z')$ be a complex random variable, where $\mathrm{Re}(Z'), \mathrm{Im}(Z') \overset{d}{=} N(0,[\underline{\sigma}^2,\overline{\sigma}^2])$. If $\mathrm{Im}(Z')$ is independent of $\mathrm{Re}(Z')$ under $\mathbb{E}$, then for a complex number $h$, $hZ'$ still satisfies both real and imaginary parts follow a G-normal distribution $N(0,[|h|^2\underline{\sigma}^2,|h|^2\overline{\sigma}^2])$.
\end{theorem}

{\bf Proof}\ \ Suppose $h=h_1+ih_2$, then $hZ'=[h_1\mathrm{Re}(Z')-h_2\mathrm{Im}(Z')]+i[h_2\mathrm{Re}(Z')+h_1\mathrm{Im}(Z')]$. Due to $\mathrm{Im}(Z')$ is an independent copy of $\mathrm{Re}(Z')$, we get that both real and imaginary parts of $hZ'$ are G-normally distributed random variables.

For the term $\mathrm{Re}(hZ')$, there holds
\begin{equation}
\begin{split}
\mathbb{E}[(\mathrm{Re}(hZ'))^2]=&\mathbb{E}[h_1^2(\mathrm{Re}(Z'))^2-2h_1h_2\mathrm{Re}(Z')\mathrm{Im}(Z')+h_2^2(\mathrm{Im}(Z'))^2]\\
\overset{(a)}{=}&\mathbb{E}[h_1^2(\mathrm{Re}(Z'))^2+h_2^2(\mathrm{Im}(Z'))^2]\\
\overset{(b)}{=}&\mathbb{E}[\mathbb{E}[h_1^2 z^2+h_2^2(\mathrm{Im}(Z'))^2|z=\mathrm{Re}(Z')]]=|h|^2\overline{\sigma}^2.
\end{split}
\end{equation}
where $(a)$ is because $\mathrm{Im}(Z')$ is independent of $\mathrm{Re}(Z')$ and $\mathrm{Im}(Z')$ has no mean uncertainty, $(b)$ is obtained by Definition \ref{definition:independent}. Similarly, we can obtain $-\mathbb{E}[-(\mathrm{Re}(hZ'))^2]=|h|^2\underline{\sigma}^2$.

For the term $\mathrm{Im}(hZ')$, the same method yields the conclusion. $\hfill\blacksquare$

The next theorem considers the case where both mean uncertainty and variance uncertainty exist simultaneously. The proof follows the same line of reasoning as that of Theorem~\ref{theorem:sufficient G-normal}. Expanding $\mathrm{Re}(hZ')$ and $\mathrm{Im}(hZ')$ in terms of the real and imaginary parts of $h$ and computing each term, together with the independence assumptions, yields the stated parameter ranges. We therefore omit the detailed derivation for brevity.

\begin{theorem}\label{theorem:sufficient G-normal and maximal}
	Let $(\Omega,\mathcal{H},\mathbb{E})$ be the sublinear expectation space. $Z'=\mathrm{Re}(Z')+i\mathrm{Im}(Z')$ be a complex random variable, where $\mathrm{Re}(Z')=M+\delta$, $ \mathrm{Im}(Z')=M_1+\delta_1$, $M,M_1 \overset{d}{=}M_{[\underline{\mu},\overline{\mu}]}$, $\delta,\delta_1 \overset{d}{=} N(0,[\underline{\sigma}^2,\overline{\sigma}^2])$. If $M_1$ is independent of $M$, $\delta_1$ is independent of $\delta$, under $\mathbb{E}$, then for a complex number $h$, $hZ'$ satisfies $\mathrm{Re}(hZ')=\widetilde{M}+\widetilde{\delta}$, $\mathrm{Im}(hZ')=\widetilde{M_1}+\widetilde{\delta_1}$, where $\widetilde{\delta},\widetilde{\delta_1} \overset{d}{=} N(0,[|h|^2\underline{\sigma}^2,|h|^2\overline{\sigma}^2])$, and
	\begin{equation}
	\widetilde{M} \overset{d}{=} M_{[\underline{\mu}(h_1^{+}+h_2^{-})-\overline{\mu}(h_1^{-}+h_2^{+}),\ \overline{\mu}(h_1^{+}+h_2^{-})-\underline{\mu}(h_1^{-}+h_2^{+}) ]},
	\end{equation}
	\begin{equation}
	\widetilde{M_1} \overset{d}{=} M_{[ \underline{\mu}(h_1^{+}+h_2^{+})-\overline{\mu}(h_1^{-}+h_2^{-}),\ \overline{\mu}(h_1^{+}+h_2^{+})-\underline{\mu}(h_1^{-}+h_2^{-}) ]}.
	\end{equation}
\end{theorem}

\section{Numerical results}\label{sec6}

In this section, we present simulation results for our proposed theorems and the detection performance of our proposed optimal detector.

Fig. \ref{fig_1} shows the effect of varying the range of distributional uncertainty of channel noise parameters on the performance of the optimal detector, plotting the maximum and minimum probabilities of error occurrence for different parameter ranges at equal SNR values. In Fig. \ref{fig_1.1}, the conclusions of Theorem \ref{theorem:AGNN detector} for the ADUN channel without mean uncertainty are presented, with channel noise $\delta$ considered individually under G-normal distributions $N(0,\Theta_1)$, $N(0,\Theta_2)$, and $N(0,\Theta_3)$, where $\Theta_1=[\underline{\sigma}^2,\overline{\sigma}^2]$, $\Theta_2=[\frac{2\overline{\sigma}^2\underline{\sigma}^2}{2\overline{\sigma}^2+\underline{\sigma}^2},2\overline{\sigma}^2]$, and $\Theta_3=[\frac{3\overline{\sigma}^2\underline{\sigma}^2}{3\overline{\sigma}^2+2\underline{\sigma}^2},3\overline{\sigma}^2]$. Fig. \ref{fig_1.2} presents the conclusions of Theorem \ref{theorem:AGNNmean detector} for the ADUN channel, with channel noise $Z$ following $M_{\Gamma_1}+N(0,\Delta_1)$, $M_{\Gamma_2}+N(0,\Delta_2)$, and $M_{\Gamma_3}+N(0,\Delta_3)$, respectively, where $\Gamma_1=[\underline{\mu},\overline{\mu}]$, $\Gamma_2=[2\underline{\mu},2\overline{\mu}]$, $\Gamma_3=[3\underline{\mu},3\overline{\mu}]$, $\Delta_1=[\underline{\sigma}^2,\overline{\sigma}^2]$, $\Delta_2=[\frac{4\overline{\sigma}^2\underline{\sigma}^2(1+\overline{\mu}-\underline{\mu})^2}{2\underline{\sigma}^2-\underline{\sigma}^2(\overline{\mu}-\underline{\mu})^2+\overline{\sigma}^2(2+\overline{\mu}-\underline{\mu})^2},2\overline{\sigma}^2]$, and $\Delta_3=[\frac{3\overline{\sigma}^2\underline{\sigma}^2(2+3\overline{\mu}-3\underline{\mu})^2}{8\underline{\sigma}^2-6\underline{\sigma}^2(\overline{\mu}-\underline{\mu})^2+3\overline{\sigma}^2(2+\overline{\mu}-\underline{\mu})^2},3\overline{\sigma}^2]$. From the figures, it can be observed that for the same SNR value, the larger the range of parameter uncertainty, the greater the maximum probability of error and the smaller the minimum probability of error. Note that, under the condition of equal SNR, the uncertain parameters of channel noise can vary across different settings. In this study, we illustrate the outcomes for one specific parameter configuration; notwithstanding, the trajectories of the resulting curves are analogous for other parameter values.

\begin{figure}[htbp]
	\centering
	\subfloat[]{\includegraphics[width=2.8in]{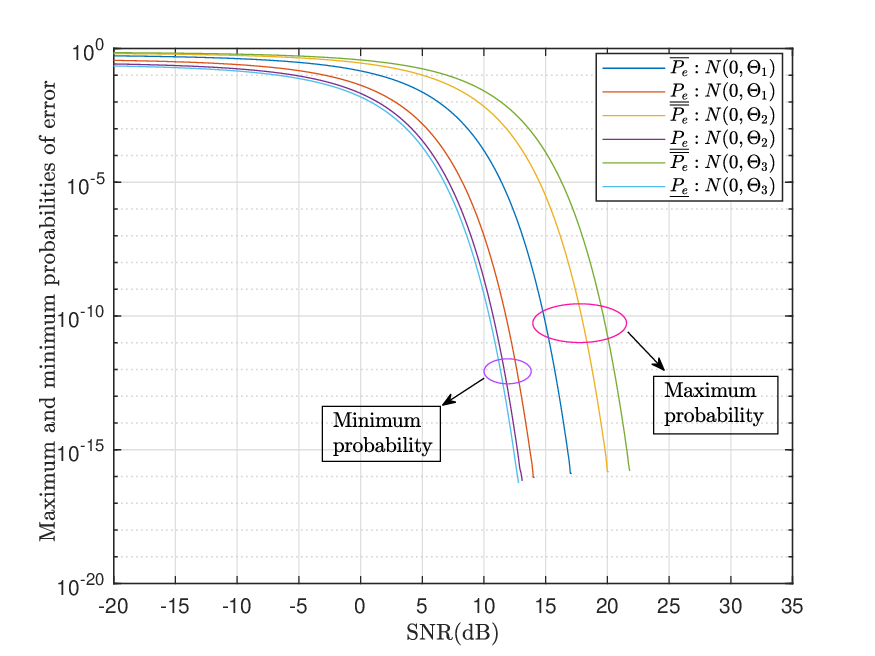}%
		\label{fig_1.1}}
	\hfil
	\subfloat[]{\includegraphics[width=2.8in]{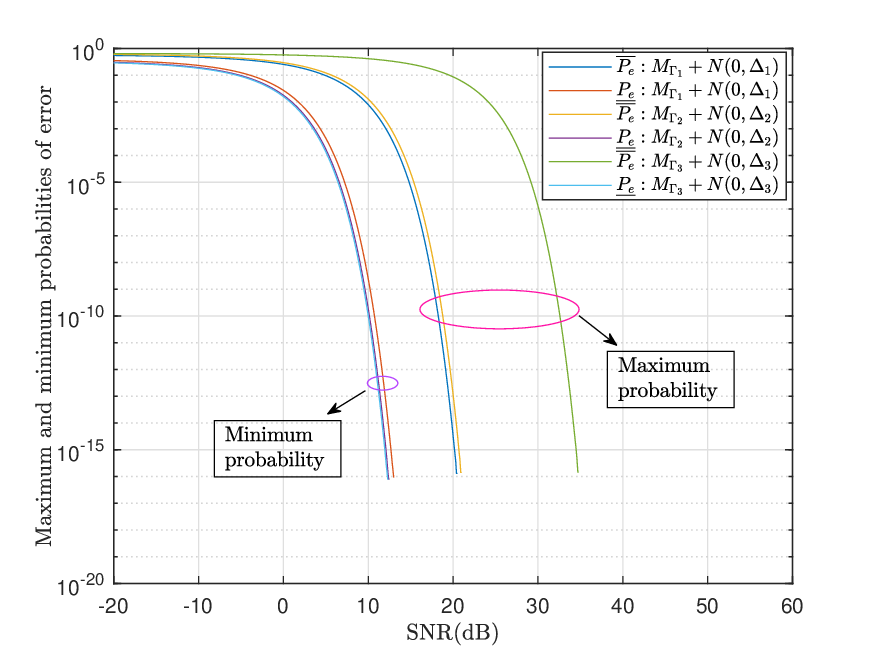}%
		\label{fig_1.2}}
	\caption{Performance of optimal threshold detectors in ADUN channel and ADUN channel without mean uncertainty under varying noise distributions.}
	\label{fig_1}
\end{figure}

Fig. \ref{fig_2} depicts the maximum and minimum probabilities of error for the ADUN channel, the ADUN channel without mean uncertainty using the optimal detector derived in this paper, and the probability of error for the AWGN channel using the minimum distance detector, at the equal SNR. The middle green line represents the results for the AWGN channel. It can be observed that the maximum and minimum probabilities derived in this study are positioned on either side of the conclusions for the AWGN channel. The blue and red lines in Fig. \ref{fig_2} represent the conclusions for the ADUN channel without mean uncertainty. We use these two lines as examples to explain the practical significance of the maximum and minimum probabilities of error. Due to the inherent uncertainty of the probability model, the transition probabilities of the channel are uncertain, making it impossible to theoretically estimate a specific error probability for a channel detector. However, if the noise data conforms to a G-normal distribution, referring to the blue line allows us to obtain a conservative estimation method. For instance, if we want the system's error probability to be below $10^{-10}$, by referring to the blue line, we can ensure it by simply increasing the SNR to $15$ dB. This conservative theoretical curve should have significant reference value for some situations that require very high accuracy, such as the application scenarios of autonomous driving and precision manufacturing. Conversely, the red line, which corresponds to the minimum error probability, represents an optimistic theoretical reference curve. To have a chance of achieving a error probability below $10^{-10}$, the red line indicates that the SNR should be set to at least $12$ dB.

From a system design perspective, the maximum and minimum error probabilities serve fundamentally different purposes. In reliability critical applications such as autonomous driving and industrial control, $\overline{P_e}$ should serve as the design benchmark. It quantifies the worst-case error rate that can occur under the worst-case noise conditions encompassed by the uncertainty model, and therefore represents an upper bound on the error probability that the system can guarantee within its declared operating range. In contrast, $\underline{P_e}$ corresponds to the error probability under the most optimistic scenario permitted by the model. It does not constitute a performance guarantee, but rather indicates the theoretical best‑case potential of the detector and can be used for diagnostics. For instance, if the long term measured error rate falls below $\underline{P_e}$, it suggests that the assumed uncertainty range of the noise is overly conservative, leading to an overestimation of the error probability; conversely, if the measured error rate exceeds $\overline{P_e}$, it may indicate that the channel has encountered noise conditions beyond those covered by the model. By jointly employing $\overline{P_e}$ and $\underline{P_e}$, one obtains a systematic framework for evaluating both the worst‑case performance guarantee and the best‑case potential of a detector under distributional uncertainty, which constitutes the practical application methodology of the theoretical results.

\begin{figure}[htbp]
	\centering
	\includegraphics[scale=0.6]{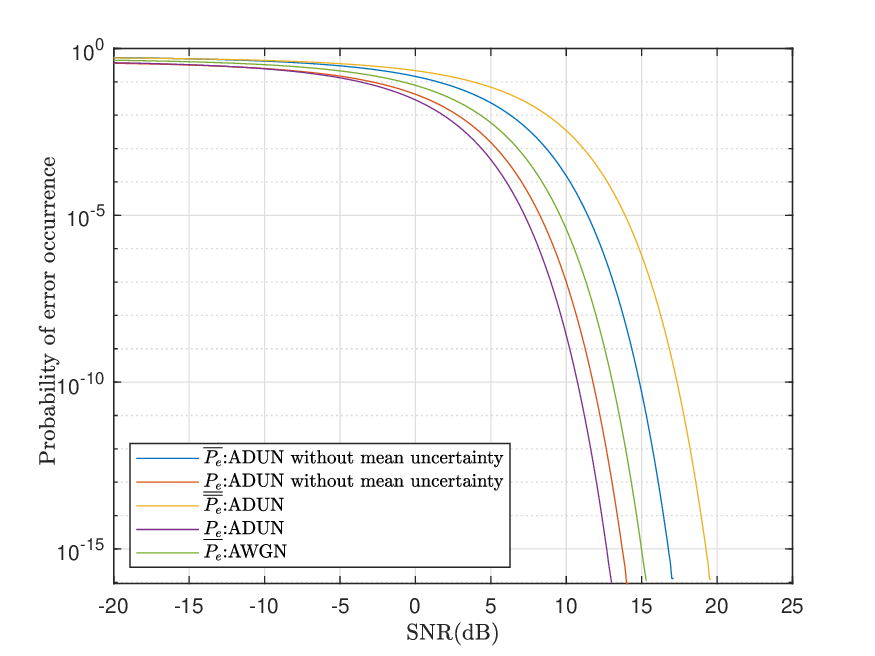}
	\caption{Comparison of optimal detector performance of ADUN channel, ADUN channel without mean uncertainty, and AWGN channel.}
	\label{fig_2}
\end{figure}

In Section \ref{subsec4.2}, we derive that if the channel noise's mean uncertainty is described as the maximum distribution $M_{[\underline{\mu},\overline{\mu}]}$, the new optimal detector is choosing $x_A$ when the received signal $Y > \frac{x_A+x_B+\overline{\mu}+\underline{\mu}}{2}$ and $x_B$ otherwise. According to the LLN under sublinear expectations, this implies that if the upper mean ($\hat{\overline{\mu}}$) and lower mean ($\hat{\underline{\mu}}$) of a set of noise data estimates are not equal, the detector can be updated using the upper and lower means. However, in classical information theory, a detector using the minimum distance criterion is employed as long as the arithmetic mean of a set of noise data is zero. In Fig. \ref{fig_3}, we conducted simulation experiments to compare the performance of the new optimal detector with that of the minimum distance detector, under conditions of uncertain noise distributions. We used binary input signal $X\in \{ -1 ,1 \}$ and simulated $10000$ sample data points for cases where the mean uncertainty is characterized by the interval $[-0.003,0.067]$ (the blue and red lines) and $[-0.021,0.074]$ (the yellow and purple lines). It can be observed from the figure that when the mean is uncertain, the performance of the new optimal detector outperforms that of the optimal minimum distance detector under classical information theory. Moreover, the larger the interval characterizing the mean uncertainty, the more significant the performance improvement becomes. This suggests that our new method is more robust to variations in noise characteristics and can adapt to a wider range of distribution uncertainty conditions. The findings underscore the potential of the new detector in applications where noise distribution is not perfectly known, offering a promising alternative to conventional detection methods.

\begin{figure}[htbp]
	\centering
	\includegraphics[scale=0.6]{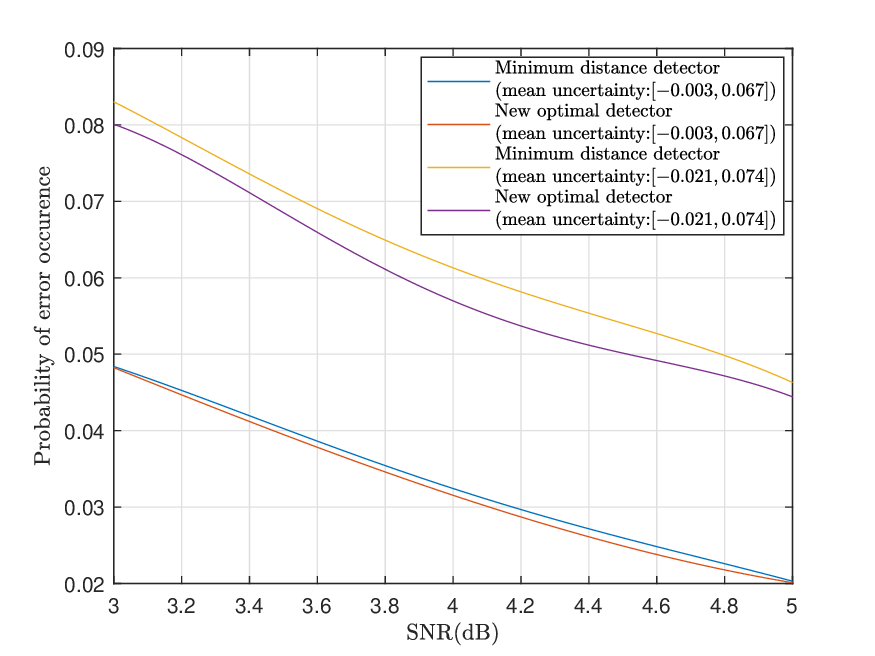}
	\caption{Simulation results of the new optimal detector and the minimum distance detector with uncertain noise distributions.}
	\label{fig_3}
\end{figure}

\section{Conclusion}\label{sec7}

In this paper, we have analyzed the fundamental detection problem for additive noise channel with uncertain noise distributions under the framework of nonlinear expectation theory. We have examined two cases involving noise uncertainty: the first scenario involves noise with both mean and variance uncertainties, while the second includes noise with no mean uncertainty but with variance uncertainty. Accordingly, we propose the concepts of the ADUN channel and ADUN channel without mean uncertainty. Based on the properties of G-normal distribution and maximum distribution from nonlinear expectation theory, we derived the family of uncertain probability distributions for output random variables given an input in these new scenarios. We then derived the optimal detectors in the sense of the nonlinear expectation optimality criterion and found that the optimal detectors are influenced by the range of mean uncertainty. Simulation experiments confirmed that the optimal detector outperforms the traditional minimum distance detector under distribution uncertainty, thereby providing experimental support for the theoretical results presented in this paper. Utilizing nonlinear expectation theory to analyze and quantify communication scenarios with uncertain probability models is a promising field. The application of nonlinear expectation theory provides a novel perspective to address the problems in information theory. We believe that this approach can reveal many previously unknown phenomena in communication system in the future.

\bigskip

\section*{Acknowledgments} 
The authors are very grateful for the help provided by Yongsheng Song in revising the manuscript.

\end{document}